\documentclass[aps,pra,groupedaddress,superscriptaddress,amsmath,amssymbol,stfloats,twocolumn]{revtex4-1}

\usepackage[latin1]{inputenc}
\usepackage{amsthm}
\usepackage{amssymb}
\usepackage{graphicx}
\usepackage{amsmath}
\usepackage{bbold}
\usepackage{bbm}
\usepackage{nonfloat}
\usepackage{color}
\usepackage[pdftex]{hyperref}
\usepackage{braket}
\usepackage{graphicx}
\usepackage{dsfont}
\usepackage{mathdots}
\usepackage{enumitem}
\usepackage{times,txfonts}

\DeclareGraphicsExtensions{.jpeg, .png, .pdf}

\newtheorem{thm}{Theorem}
\newtheorem*{thm*}{Theorem}
\newtheorem{prop}[thm]{Proposition}
\newtheorem*{prop*}{Proposition}
\newtheorem{lemma}[thm]{Lemma}
\newtheorem*{lemma*}{Lemma}
\newtheorem{cor}[thm]{Corollary}
\newtheorem*{cor*}{Corollary}

\newtheorem*{ex*}{Example}

\newtheorem*{cj*}{Conjecture}

\newtheorem*{Def*}{Definition}

\theoremstyle{definition}
\newtheorem*{rem}{Remark}

\newcommand{\bq}{\begin{equation*}}
\newcommand{\be}{\begin{equation}}
\newcommand{\eq}{\end{equation*}}
\newcommand{\ee}{\end{equation}}
\newcommand{\bmu}{\begin{multline*}}
\newcommand{\emu}{\end{multline*}}
\newcommand{\ban}{\begin{align*}}
\newcommand{\bal}{\begin{align}}
\newcommand{\ean}{\end{align*}}
\newcommand{\eal}{\end{align}}
\newcommand{\Tr}{\text{Tr}\,}

\newcommand{\red}[1]{{#1}} 

\begin{document}

\title{Schur complement inequalities for covariance matrices and monogamy of quantum correlations}

\author{Ludovico Lami}
\affiliation{F\'{\i}sica Te\`{o}rica: Informaci\'{o} i Fen\`{o}mens Qu\`{a}ntics, Departament de F\'{i}sica, Universitat Aut\`{o}noma de Barcelona, ES-08193 Bellaterra (Barcelona), Spain}

\author{Christoph Hirche}
\affiliation{F\'{\i}sica Te\`{o}rica: Informaci\'{o} i Fen\`{o}mens Qu\`{a}ntics, Departament de F\'{i}sica, Universitat Aut\`{o}noma de Barcelona, ES-08193 Bellaterra (Barcelona), Spain}

\author{Gerardo Adesso}
\affiliation{Centre for the Mathematics and Theoretical Physics of Quantum Non-Equilibrium Dynamics, School of Mathematical Sciences, The University of Nottingham,
University Park, Nottingham NG7 2RD, United Kingdom}

\author{Andreas Winter}
\affiliation{F\'{\i}sica Te\`{o}rica: Informaci\'{o} i Fen\`{o}mens Qu\`{a}ntics, Departament de F\'{i}sica, Universitat Aut\`{o}noma de Barcelona, ES-08193 Bellaterra (Barcelona), Spain}
\affiliation{ICREA -- Instituci\'o Catalana de Recerca i Estudis Avan\c{c}ats, Pg. Lluis Companys 23, ES-08010 Barcelona, Spain}

\begin{abstract}
We derive fundamental constraints for the Schur complement of positive matrices, which provide an operator strengthening to recently established information inequalities for quantum covariance matrices, including strong subadditivity. This allows us to prove general results on the monogamy of entanglement and steering quantifiers in continuous variable systems with an arbitrary number of modes per party. A powerful hierarchical relation for correlation measures based on the log-determinant of covariance matrices is further established for all Gaussian states, which has no counterpart among quantities based on the conventional von Neumann entropy.
\end{abstract}
\date{July 16, 2016}

\maketitle

Quantum correlations embody the true departure of quantum mechanics from ``classical lines of thought'' \cite{schr}. In recent years, the mathematical development of quantum information theory and the experimental progress in controlling quantum systems have greatly advanced our physical understanding of quantum correlations. Different incarnations of quantum correlations, such as nonlocality, steering, entanglement, and discord, can arise in generally mixed multipartite states \cite{ABC}, and can all be exploited to achieve enhancements in information processing tasks over purely classical scenarios \cite{chuaniels}. On the other hand, quantum correlations also come with fundamental limitations not affecting classical ones, such as their {\it monogamy}, that is, the fact that quantum correlations cannot be freely shared across many subsystems \cite{terhal_2004,coffman_2000,koashi_2004,osborne_2006,hiroshima_2007,strongmono,toner_2009,streltsov_2012,regula_2014,AdessoSerafini,siewert_2015,MonSteer,lancien_2016}. Even such a limitation has useful applications, as it leads to the unconditional security of quantum key distribution \cite{devetak_2005}. Carefully identifying structural similarities and key differences between classical correlations and different types of quantum correlations is a paramount step to assess the resource power of the latter ones.

Interestingly, there are trademark quantum systems whose mathematical description is as simple as that of their classical counterparts. Such is the case for systems of (quantum) harmonic oscillators, e.g.~modes of the electromagnetic field, whose ground and thermal-equilibrium states belong to the special set of Gaussian states \cite{AdessoReview}. The study of these states and of the operations which preserve their Gaussianity is entirely ascribed to the characterization of covariance matrices (CMs) and their transformations using methods of linear algebra and symplectic geometry, which are widely applied in {\it classical} mechanics \cite{arnold}. Yet Gaussian states and channels realize paradigmatic platforms for continuous variable  {\it quantum} information processing \cite{cvbook}, and have been used to successfully demonstrate unconditional quantum teleportation in optical and atomic domains \cite{furuscience,spincoherent,fernnp13}, quantum cryptography with coherent states \cite{Grosshans2003}, and sub-shot-noise interferometry in gravitational wave detectors \cite{SGravi1,SGravi2}, among others. CMs also encode useful information on more general, non-Gaussian states \cite{Carles,AdessoNG}, leading to easily testable qualitative criteria and quantitative lower bounds for their non-classical properties \cite{extra}. One then wonders to what extent the description of  quantum correlations (in Gaussian states and beyond) can be advanced by further developing suitable classical methods.

{In this Letter we establish a collection of results for the Schur complement of a CM \red{--- a submatrix encoding conditional covariance, of use in linear algebra, numerical methods, probability and statistics ---} which bear a direct impact on  the quantitative characterization of various forms of quantum correlations in continuous variable systems, and in turn on their usefulness for quantum technologies.} Our analysis is inspired by recent works \cite{AdessoSerafini,Gross,Adesso}, in which an inequality sharing the same formal structure as the strong subadditivity of entropy was obtained, by purely algebraic methods, 
for the log-determinant of positive semidefinite matrices $V_{ABC} \geq 0$:
\begin{equation}\label{log det ineq 2}  \log\det V_{ABC} + \log\det V_{C}\, \leq\, \log\det V_{AC} +\log\det V_{BC}\,.  \end{equation}
If one identifies $V_{ABC}$ with the CM of a $(n_A+n_B+n_C)$-mode tripartite {\it quantum} system, which requires the extra condition
\begin{equation}\label{bonafide}
V_{ABC} + i \Omega_{ABC} \geq 0\,,
\end{equation}
encapsulating the uncertainty principle \cite{Simon94} (with $\Omega_{ABC} = \Omega^{\oplus (n_A+n_B+n_C)}$, and $\Omega =  {{\ 0\ \ 1}\choose{-1\ 0}}$  being the symplectic form),
then the scalar inequality (\ref{log det ineq 2}) has relevant implications, yielding alternative quantifiers of correlations \cite{AdessoSerafini}, monogamy constraints for Gaussian entanglement \cite{AdessoSerafini}, and limitations for joint steering of single-mode states in a multipartite scenario \cite{Adesso,MonSteer}. Here we show, {\it inter alia},
that such an inequality admits a powerful {\it operator} strengthening directly at the level of CMs, which allows us to substantially generalize the monogamy results of \cite{AdessoSerafini,Adesso,MonSteer} to multimode Gaussian or non-Gaussian states with any number of modes per party. For Gaussian states, we further establish a fundamental hierarchy for bipartite correlations based on the log-determinant, which does not hold for the standard entropy \cite{LiLuo}. In what follows, we first present our general results for Schur complements, and later explore their consequences in the quantum domain. We refer to a CM as any symmetric and positive semidefinite matrix, and to a {\it quantum} CM as one additionally obeying (\ref{bonafide}).

\paragraph*{Schur complement inequalities.}
The Schur complement is an operation that takes as input a $n\times n$ matrix $M$ and one of its $k\times k$ principal submatrices $A \sqsubset M$ (the shorthand $X\sqsubset Y$ means $X$ is a square submatrix of $Y$), and outputs a $(n-k)\times (n-k)$ matrix $M/A$. For a CM $M$ written in block form as
\begin{equation} M\, =\, \begin{pmatrix} A && X \\ X^T && B \end{pmatrix}\, , \label{M} \end{equation}
one defines
the Schur complement of $A$ in $M$ as $M/A =  B - X^T A^{-1} X$, and analogously  $M/B=A-XB^{-1} X^T$. The inverses here are taken on the support. We now list some useful properties of the Schur complement \cite{HornJohnson1}. (i) Determinant factorization:
$\det M = \det A \det(M/A)$;
(ii) Inversion formula:
\begin{equation} M^{-1}\ =\ \begin{pmatrix} A^{-1} + A^{-1} X (M/A)^{-1} X^T A^{-1} && -A^{-1} X (M/A)^{-1} \\[0.7ex] -(M/A)^{-1} X^T A^{-1} && (M/A)^{-1} \end{pmatrix}\, , \label{inv} \end{equation}
with the easy corollary $M^{-1} \big/ (M/A)^{-1}\ =\ A^{-1}$; (iii) Congruence invariance: conformally to the partition in \eqref{M}, we have
$\left(\begin{smallmatrix} N_1 & \\ & N_2 \end{smallmatrix}\right) M \left(\begin{smallmatrix} N_1^T & \\ & N_2^T \end{smallmatrix}\right) \Big/ N_1^T A N_1  \geq N_2 \left( M/A \right) N_2^T$, for all $N_1, N_2$, with equality if $N_1$ is invertible;  (iv) Quotient property: if $A\sqsubset M$ and $A_1\sqsubset A$, then $A/A_1\sqsubset M/A_1$ and moreover
$M/A = (M/A_1) \big/ (A/A_1)$; (v) Variational characterization: $M/A\ =\ \max\,\big\{ W:\ M\geq 0 \oplus W \big\}$ (i.e., the latter set of matrices $\{ W \}$ has a unique supremum given by  $M/A$), which means in particular that  $M\mapsto M/A$  is monotonically increasing and concave, while $M\mapsto (M/A)^{-1}$ is decreasing and convex.

Interestingly, it follows from the latter property (v) that $\log\det(M/A)=\Tr \log(M/A)$ is concave in $M$ thanks to the operator concavity of the logarithm. This leads to a simple proof of the central finding of \cite{Adesso}, i.e. the inequality
\begin{equation} \label{log det ineq}  \log\det V_{AC} +\log\det V_{BC} - \log\det V_{A} - \log\det V_{B} \geq 0\, ,  \end{equation}
valid for any quantum CM $V_{ABC}$. This is obtained by rewriting the left-hand side as $\log\det (V_{AC}/V_A) + \log\det (V_{BC}/V_B)$, which is a concave function of $V_{ABC}$, and by noticing that (\ref{log det ineq}) is saturated for pure states. Observe that (\ref{log det ineq 2}) and (\ref{log det ineq}) are equivalent expressions of strong subadditivity for the log-determinant, and can be converted into each other by `purifying' $V_{ABC}$ into a {\it symplectic} CM $V_{ABCD} \sqsupset V_{ABC}$ (describing a pure Gaussian state of  $ABCD$) \cite{Adesso}, where we recall that a matrix $S$ is symplectic if $S \Omega S^T = \Omega$, which implies $\det S=1$ \cite{Arvind95}.

\red{This suggests that the  Schur complement of CMs can define a natural notion of conditional covariance, as previously noted for classical Gaussian variables \cite{simoprl}.}
Hence we will fix $M \equiv V_{AB} \geq 0$ and study the Schur complement  $V_{AB}/V_B$, thereby proving that many well-known properties of the standard conditional entropy $H(A|B)=H(AB)-H(B)$, where $H$ denotes respectively Shannon or von Neumann entropy for a classical or quantum system,
have a straightforward equivalent within this framework.


We start by recalling that a canonical formulation of strong subadditivity in classical and quantum information theory is $H(A|BC) \leq H(A|C)$, i.e.~partial trace on the conditioning system \red{increases} the conditional entropy \cite{ArakiLieb,WehrlR,QuantumSSA,QuantumSSA2}. Guided by our formal analogy, our first result is thus a generalization of \eqref{log det ineq 2}.

\begin{thm}[Partial trace in the denominator \red{increases} Schur complement] \label{pt incr sch}
If $V_{ABC} \geq 0$ is any tripartite CM, then
\begin{equation} V_{ABC} / V_{BC}\ \leq\ V_{AC} / V_C\, . \label{INEQ 1} \end{equation}
\end{thm}

\begin{proof}
Since $V_{ABC}\geq W_A \oplus 0_{BC}$ implies $V_{AC}\geq W_A\oplus 0_C$, employing property (v) we find
$V_{ABC} / V_{BC}\ =\ \max\big\{ W_A:\, V_{ABC}\geq W_A\oplus 0_{BC} \big\}\ \leq\ \max\big\{ W_A:\, V_{AC}\geq W_A\oplus 0_C \big\}\ =\ V_{AC} / V_C$.
\end{proof}

Clearly, taking the determinant of \eqref{INEQ 1} and applying the factorization property (i) of the Schur complement yields \eqref{log det ineq 2} immediately.
Notice further that the invariance of $V_{AB}/V_B$ under symplectic operations on $B$ (implied by the congruence property (iii)) and its monotonicity under partial trace, suffice to guarantee its monotonicity under general deterministic (i.e.~trace-preserving) Gaussian channels  $\Gamma_B$ on $B$:
$(\mathds{1}_A\oplus \Gamma_B)(V_{AB}) \,\big/\, \Gamma_B (V_B)= \big( S_{BC}\,( V_{AB}\oplus \sigma_C)\, S_{BC}^T \big)_{AB} \Big/ \big( S_{BC}\, (V_B\oplus \sigma_C)\, S_{BC}^T \big)_B \geq
\big( S_{BC}\, (V_{AB}\oplus \sigma_C)\, S_{BC}^T \big) \Big/ \big( S_{BC}\, (V_B\oplus \sigma_C)\, S_{BC}^T \big) = (V_{AB}\oplus \sigma_C) \big/ (V_B\oplus \sigma_C) = V_{AB} / V_B$. But there is more: perhaps surprisingly,  the Schur complement is also  monotonically increasing under general non-deterministic {\it classical} (i.e.~non quantum-limited) Gaussian operations on $B$. We recall that any such map acts at the level of CMs as \cite{nogo1,nogo2,nogo3}
\begin{equation}\label{CP gauss}
\Gamma_{B\rightarrow B'}:\ V_B\longmapsto \gamma_{B'} - \delta_{BB'}^T\,\left({\gamma_B+V_B}\right)^{-1}\,\delta_{BB'}\, ,
\end{equation}
where $\gamma_{BB'}=\left(\begin{smallmatrix} \gamma_B & \delta_{BB'} \\ \delta_{BB'}^T & \gamma_{B'} \end{smallmatrix}\right)>0$ is a positive matrix pertaining to a bipartite system $BB'$. If $\gamma_{BB'}$ is also a valid quantum CM obeying (\ref{bonafide}), then \eqref{CP gauss} corresponds to a (non-deterministic) completely positive Gaussian channel, but this restricting hypothesis plays no role in stating the following general result.

\begin{thm}[Classical Gaussian maps in the denominator increase Schur complement] \label{CP incr sch}
If $\Gamma_{B\rightarrow B'}$ is a non-deterministic classical Gaussian map as in \eqref{CP gauss}, with $\left(\begin{smallmatrix} \gamma_B & \gamma_{BB'} \\ \gamma_{BB'}^T & \gamma_{B'} \end{smallmatrix}\right)>0$, then
\begin{equation*}
\Gamma_{B\rightarrow B'}(V_{AB}) \big/ \Gamma_{B\rightarrow B'}(V_B)\ \geq\ V_{AB} / V_B\, .
\end{equation*}
\end{thm}

\begin{proof}
Observing that \eqref{CP gauss} can be rewritten as $\Gamma_{B\rightarrow B'}:\ V_B\longmapsto (\gamma_{BB'}+V_B) \big/ (\gamma_B+V_B)$,
we obtain: $\Gamma_{B\rightarrow B'}(V_{AB}) \big/ \Gamma_{B\rightarrow B'}(V_B) =  \big((\gamma_{BB'}+V_{AB}) / (\gamma_B+V_B)\big) \Big/ \big((\gamma_{BB'}+V_B) / (\gamma_B+V_B) \big) =
  (\gamma_{BB'}+V_{AB})  \big/ (\gamma_{BB'}+V_B) \geq V_{AB} / V_B$,
where we used  property (iv) together with the bound $\left.\left(\begin{smallmatrix} A & X \\ X^T & B+\sigma \end{smallmatrix}\right) \Big/ (B+\sigma)\ \geq\ \left(\begin{smallmatrix} A & X \\ X^T & B \end{smallmatrix}\right) \Big/ B\right.$.
\end{proof}

Next, we would like to obtain from \eqref{INEQ 1} an operator generalization of \eqref{log det ineq} by applying the symplectic purification trick: this requires a certain amount of work. From now on, we will always assume that the $V$ matrices are bona fide quantum CMs obeying  (\ref{bonafide}). We first note that if a bipartite quantum CM $V_{AB}$ is symplectic, then $V_{AB}^{-1} = \Omega_{A}^T V_{AB}^T \Omega_{AB} = \Omega_{AB}^T V_{AB} \Omega_{AB}$, which by comparison with (\ref{inv}) yields $V_{AB}/V_A=\Omega_B^T V_B^{-1} \Omega_B$. In conjunction with property (iv), this implies
\begin{equation}\label{lemma2}
V_{ABC} \mbox{ is symplectic \ \ $\Rightarrow$ \ \ } V_{AB}/V_B = \Omega_A^T (V_{AC}/V_C)^{-1} \Omega_A\,.
\end{equation}
We then get the following for any tripartite quantum system.
\begin{thm}[Schur complement of quantum CMs is monogamous] \label{sch compl mon}
If $V_{ABC}\geq i \Omega_{ABC}$ is any tripartite quantum CM, then
\begin{equation} V_{AC} / V_A\ \geq\ \Omega_C^T (V_{BC} / V_B)^{-1} \Omega_C\, . \label{INEQ 2} \end{equation}
\end{thm}

\begin{proof}
Consider a symplectic purification $V_{ABCD}$ of the system $ABC$. Applying first \eqref{INEQ 1} and then (\ref{lemma2}) yields  \eqref{INEQ 2}: $V_{AC}/V_A\ \geq\ V_{ACD}/V_{AD}\ =\ \Omega_C^T (V_{BC}/V_B)^{-1} \Omega_C$.
Alternatively, observe that the difference between right- and left-hand side of \eqref{INEQ 2} is concave in $V_{ABC}$ (as $V_{AC}/V_A$ is concave and $(V_{BC}/V_B)^{-1}$ is convex), and it vanishes on symplectic CMs by (\ref{lemma2}).
\end{proof}

We remark that the operator inequalities (\ref{INEQ 1}) and (\ref{INEQ 2}) are significantly stronger than the scalar ones (\ref{log det ineq 2}) and (\ref{log det ineq}) reported in \cite{AdessoSerafini,Gross,Adesso}, as the former establish algebraic limitations directly at the  level of CMs, in a similar spirit to the marginal problem \cite{Tyc2008},
for arbitrary multipartite states.
Equipped with these powerful tools, we proceed to investigate applications to quantum correlations, namely steering and entanglement.

\paragraph*{Gaussian steerability and its monogamy.}
Consider a $n$-mode continuous variable quantum system, and denote by $\nu_i(A)$ the $i$--th smallest symplectic eigenvalue of a positive definite CM $0<A=A^T\in \mathcal{M}_{2n}(\mathds{R})$. We define the two functions
\begin{equation} g_\pm (A) =  {\sum}_{i=1}^n\, \max\,\big\{\pm \log \nu_i(A),\, 0\big\}\, . \label{func g} \end{equation}

The function $g_-$ finds many applications in continuous variable quantum information. For instance, the logarithmic negativity \cite{VidalWerner,plenioprl} of a bipartite state $\rho_{AB}$, defined as  $E_N(\rho_{AB}) =  \log \|\rho_{AB}^{\text{\reflectbox{$\Gamma$}}} \|_1$ (where \reflectbox{$\Gamma$} denotes partial transposition), takes the form $E_N(\rho_{AB})=g_-(\tilde{V}_{AB})$ if $\rho_{AB}$ is a Gaussian state with quantum CM $V_{AB}$; here, the partial transpose of the CM is given by $\tilde{V}_{AB}=\Theta V_{AB} \Theta$, with $\Theta=\left(\begin{smallmatrix} \mathds{1} & \\ & -\mathds{1} \end{smallmatrix}\right)_A\oplus \mathds{1}_B$. Furthermore, a quantitative measure of Gaussian steerability (i.e., steerability by Gaussian measurements) has been recently introduced for any state $\rho_{AB}$ with quantum CM $V_{AB}$ \cite{steerability}, that takes the form
\begin{equation}
\mathcal{G}(A\rangle B)_V = g_-(V_{AB}/V_A)\,,
\label{G steer}
\end{equation}
\red{in the case of party $A$ steering party $B$. Notice that $\mathcal{G}(A\rangle B)_V>0$ is necessary and sufficient for ``$A$ to $B$'' steerability of a Gaussian state with quantum CM $V_{AB}$ by means of Gaussian measurements on $A$ \cite{steerability,Wiseman}, but is only sufficient if either the state \cite{JOSAB} or the measurements \cite{NoGauss1,NoGauss2} are non-Gaussian.}

The functions $g_\pm$ have useful properties (see \cite{epaps}\nocite{sympsplemma,Gosson,logconcave,Bhatia,williamson,BhatiaMatrix} for details):
$g_\pm(A)=g_\pm(SAS^T)$ for all symplectic $S$, $g_\pm (A^{-1}) = g_\mp (A)$, $g_+(A)-g_-(A)=\frac{1}{2}\,\log\det A$, $g_\pm (A\oplus B)=g_\pm (A) + g_\pm (B)$, $g_-(A)$ is monotonically decreasing and convex in $A$, while $g_+(A)$ is monotonically increasing but neither convex nor concave in $A$, and finally  $g_-$ is superadditive in the subsystems,
\begin{equation}\label{dec red g- eq}
g_-(V_{AB}) \geq g_-(V_A) + g_-(V_B)\,.
\end{equation}

Based on these facts, whose proof relies on recent advances in the study of symplectic eigenvalues \cite{bhatia15}, we can prove fully general properties of the steerability measure \eqref{G steer}, extending the results of \cite{steerability} where these properties were only proven in the special case of one-mode steered subsystem ($n_B=1$).

\begin{thm}[Properties of Gaussian steerability] \label{G prop}
\red{
(1)  ${\cal G}(A\rangle B)_V$ is convex and decreasing in the CM $V_{AB}$;
(2) ${\cal G}(A\rangle B)$ is additive under tensor products, i.e.~under direct sums of CMs,
${\cal G}(A_1 A_2\rangle B_1 B_2)_{V_{A_1B_1}\oplus W_{A_2 B_2}} =  \mathcal{G}(A_1\rangle B_1)_{V_{A_1B_1}} +  \mathcal{G}(A_2\rangle B_2)_{W_{A_2 B_2}}$;
(3) for arbitrary states, ${\cal G}(A\rangle B)$ is decreasing under general, non-deterministic Gaussian maps on the steering party $A$;
(4) for Gaussian states, ${\cal G}(A\rangle B)$ is decreasing under general, non-deterministic Gaussian maps on the steered party $B$;
(5) for any quantum CM $V_{ABC}$,  it holds ${\cal G}(A\rangle C)_V  \leq  g_+(V_{BC}/V_B)$.}
\end{thm}
\begin{proof} See Appendix \cite{epaps} for detailed proofs. \end{proof}

\red{Theorem~\ref{G prop} establishes ${\cal G}(A \rangle B)_V$ as a convex monotone for arbitrary Gaussian states with quantum CM $V_{AB}$ under arbitrary local Gaussian operations on either the steering or the steered parties, hence fully validating it within the Gaussian subtheory of the recently formulated resource theory of steering \cite{resource}.} Moreover, our framework allows us to address the general problem of the monogamy of ${\cal G}(A \rangle B)$ for {\it arbitrary} (Gaussian or not) multimode states. For a state with quantum CM $V_{AB_1\ldots B_k}$, consider the following inequalities
\begin{eqnarray}
\mathcal{G}(A\rangle B_1\ldots B_k) &\geq& {\sum}_{j=1}^k \mathcal{G}(A\rangle B_j) \label{mon steer 1}\,,\\
\mathcal{G}(B_1\ldots B_k \rangle A) &\geq& {\sum}_{j=1}^k \mathcal{G}(B_j\rangle A) \label{mon steer 2}\,.
\end{eqnarray}
In a very recent study \cite{MonSteer}, both inequalities were proven in the special case of a  $(k+1)$-mode system with one single mode per party, i.e., $n_A=n_{B_j}=1$ ($j=1,\ldots,k$). We now show that  only one of these constraints holds in full generality.

\begin{thm}[Monogamy of Gaussian steerability] \label{G mono}
(a) Ineq.~(\ref{mon steer 1}) holds for any multimode quantum CM $V_{A B_1 \ldots B_k}$. (b) Ineq.~(\ref{mon steer 2}) holds for any multimode quantum CM $V_{A B_1 \ldots B_k}$ such that either $A$ comprises a single mode ($n_A=1$), or $V_{A B_1 \ldots B_k}$ is symplectic ($\det V_{A B_1 \ldots B_k}=1$), but can be violated otherwise.
\end{thm}
\begin{proof} (a) It suffices to prove the inequality $\mathcal{G}(A\rangle BC)\, \geq\, \mathcal{G}(A\rangle B) + \mathcal{G}(A\rangle C)$ for a tripartite quantum CM $V_{ABC}$, as (\ref{mon steer 1}) would follow by iteration. Observe that $V_{AB}/V_A$ and $V_{AC}/V_A$ form the diagonal blocks of the bipartite matrix $V_{ABC}/V_A$. Applying \eqref{dec red g- eq} one thus obtains $\mathcal{G}(A\rangle BC)_V = g_-(V_{ABC}/V_A) \geq g_-(V_{AB}/V_A) + g_-(V_{AC}/V_A) = \mathcal{G}(A\rangle B)_V + \mathcal{G}(A\rangle BC)_V$, concluding the proof. (b) For the case $n_A=1$ with $n_{B_j}$ arbitrary, one exploits the fact that only one term $\mathcal{G}(B_j\rangle A)$ in the right-hand side of (\ref{mon steer 2}) can be nonzero, due to the impossibility of jointly steering a single mode by Gaussian measurements as implied by (\ref{log det ineq}) \cite{Adesso}, combined with the monotonicity of $\mathcal{G}(B_1\ldots B_k \rangle A)$ under partial traces on the steering party as implied by Theorem~\ref{G prop}. Finally, the case when $V_{A B_1 \ldots B_k}$ is symplectic, i.e.~corresponding to a {\it pure} multimode Gaussian state, follows from the forthcoming Corollary~\ref{E mono} (see \cite{epaps} for further details).
\end{proof}

The Gaussian steerability is thus not monogamous with respect to a common steered party $A$ when the latter is made of two or more modes, with violations of (\ref{mon steer 2}) existing already in a tripartite setting ($k=2$) with $n_{B_1}=n_{B_2}=1$ and $n_A=2$; a counterexample is reported in \cite{epaps}. 
What is truly monogamous is the log-determinant of the Schur complement, which only happens to coincide with the function $g_-$ when $n_A=1$.

\paragraph*{Gaussian entanglement and correlations hierarchy.}
In this last section, we specialize our attention to Gaussian states \cite{AdessoReview}. The R\'enyi-2 entropy of a $n$-mode Gaussian state $\rho$ with quantum CM $V$ is given by half the log-determinant of the latter, ${\cal S}_2(\rho) = -\log \Tr \rho^2 = \frac12 \log \det V$, and is equivalent (up to an additive constant) to the classical Boltzmann--Shannon entropy of the Wigner distribution of  $\rho$ \cite{AdessoSerafini}. Owing to the strong subadditivity inequality (\ref{log det ineq 2}), one can define faithful R\'enyi-2 measures of total correlations ${\cal I}_2$ and entanglement ${\cal E}_2$ for a bipartite Gaussian state with quantum CM $V_{AB}$ \cite{AdessoSerafini}, given by
\begin{eqnarray}
{\cal I}_2(A:B)_V &=&  \mbox{$\frac{1}{2}\, \log \frac{\det V_A \det V_B}{\det V_{AB}}$}\,, \label{I2}\\
\mathcal{E}_2(A:B)_V & =& \inf_{\gamma_{AB}\, \text{pure}:\ \gamma_{AB}\,\leq\, V_{AB}}\, \mbox{$\frac{1}{2}\, \log\det \gamma_A$}\, , \label{E2}
\end{eqnarray}
where the infimum over pure Gaussian states with symplectic CM $\gamma_{AB}$ in (\ref{E2}) amounts to the Gaussian convex roof \cite{GEOF}.

In \cite{LiLuo}, the inequality ${\cal I}\geq 2\mathcal{E}$ is identified as a fundamental postulate for a consistent theory of quantum versus classical correlations in bipartite systems, for an arbitrary measure of entanglement $\mathcal{E}$ and of total correlations ${\cal I}$. \red{This follows from the fact that for pure states classical and quantum correlations are equal and add up to the total correlations \cite{Groisman2005}, while for mixed states classical correlations are intuitively expected to exceed quantum ones, which include entanglement \cite{Henderson2001,Groisman2005,LiLuo}.} However, such a relation can already be violated for two-qubit states (Werner states) when $\mathcal{E}$ is the entanglement of formation defined via the usual von Neumann entropy \cite{horodecki_2009}, and ${\cal I}$ the corresponding mutual information. In larger dimensions it may even happen that ${\cal I}<\mathcal{E}$  \cite{hayden06}, undermining the interpretation of the entanglement of formation as just a fraction of total correlations.

Here we show that ${\cal I}_2\geq 2 \mathcal{E}_2$ \emph{does hold} for Gaussian states of arbitrarily many modes using the R\'enyi-2 quantifiers \footnote{This holds in fact for all R\'enyi-$\alpha$ quantifiers with $\alpha \geq 2$ \cite{epaps}.}.
\begin{thm}[Gaussian R\'enyi-2 correlations hierarchy]\label{I2E}
Let $AB$ be in an arbitrary Gaussian quantum state. Then
\begin{equation}
\mbox{$\frac12 {\cal I}_2(A:B) \geq \mathcal{E}_2(A:B) \geq {\cal G}(A \rangle B)$}\, . \label{I>2E}
\end{equation}
If $AB$ is in a pure Gaussian state, all the above three quantities coincide with the reduced R\'enyi-2 entropy $\frac{1}{2}\log\det V_A$.
\end{thm}
\begin{proof}
The rightmost inequality is a corollary of Theorem~\ref{G prop}. The leftmost inequality admits a neat proof that makes use of the geometric mean $M\#N\equiv M^{1/2} \big( M^{-1/2} N M^{-1/2} \big)^{1/2} M^{1/2}$ between positive matrices $M,N$ \cite{Ando}. The key step is that, for any quantum CM $V_{AB}$ obeying (\ref{bonafide}), the matrix $\gamma^{\#}_{AB} = V_{AB}\#(\Omega_{AB} V_{AB}^{-1}\Omega_{AB}^T)$ is the quantum CM of a pure Gaussian state obeying $\gamma^{\#}_{AB} \leq V_{AB}$; using it as an ansatz in (\ref{E2}) and exploiting Theorem 3 in \cite{Ando} (see \cite{epaps} for full details) one shows that $\mathcal{E}_2(A:B)_V  \leq \frac{1}{2} \log\det \gamma^{\#}_A \leq \frac{1}{2} {\cal I}_2(A:B)_V$.
\end{proof}
Remarkably, this proves that the involved measures quantitatively capture the general hierarchy of correlations \cite{ABC} in arbitrary Gaussian states \cite{AdessoReview}: the Gaussian steerability is generally smaller than the entanglement degree, which accounts for a portion of quantum correlations up to half the total ones.

A crucial consequence of Theorem~\ref{I2E} is that the R\'enyi-2 measure of entanglement can now be proven monogamous for {\it arbitrary} Gaussian states with any number of modes per party.

\begin{cor}[Monogamy of Gaussian R\'enyi-2 entanglement] \label{E mono}
The Gaussian R\'enyi-2  entanglement measure (\ref{E2}) is monogamous for any multipartite Gaussian state, i.e.
\begin{equation}
\mathcal{E}_2(A: B_1\ldots B_k) \geq {\sum}_{j=1}^k \mathcal{E}_2(A:B_j) \label{mon E2}\,.
\end{equation}
\end{cor}
\begin{proof}  It suffices again to prove that $\mathcal{E}_2(A: BC)_V \geq \mathcal{E}_2(A: B)_V + \mathcal{E}_2(A: C)_V$ holds for any tripartite quantum CM $V_{ABC}$. Take the pure state with symplectic CM $\gamma_{ABC}\leq V_{ABC}$ that saturates the infimum in the definition of $\mathcal{E}_2(A: BC)$ and notice that
$\mathcal{E}_2(A: BC) = \frac{1}{2} \log\det \gamma_A = \frac{1}{2} {\cal I}_2(A:BC)_\gamma = \frac{1}{2} {\cal I}_2(A:B)_\gamma  + \frac{1}{2} {\cal I}_2(A:C)_\gamma$, where the last equality holds specifically for pure states. Applying \eqref{I>2E} to each of the two rightmost addends yields $\mathcal{E}_2(A: BC)_V \geq \, \mathcal{E}_2(A: B)_\gamma\, +$ $\mathcal{E}_2(A: C)_\gamma \geq \mathcal{E}_2(A: B)_V + \mathcal{E}_2(A: C)_V$, where the last step follows as ${\cal E}_2$ is a decreasing function of the CM.
\end{proof}

Corollary~\ref{E mono} yields the {\it most general} result to date regarding quantitative monogamy of continuous variable entanglement \cite{ourreview,AdessoReview}, as all previous proofs (for the R\'enyi-2 measure \cite{AdessoSerafini} or other quantifiers \cite{hiroshima_2007,strongmono}) were restricted to the special case of {\it one} mode per party. Combining (\ref{I>2E}) with (\ref{mon E2}), one also proves (\ref{mon steer 2}) for all pure Gaussian states, i.e., for all symplectic quantum CMs $V_{A B_1 \ldots B_k}$, as claimed in Theorem~\ref{G mono}(b).

\paragraph*{Conclusions.} We have derived fundamental inequalities for the Schur complement of positive semidefinite matrices and explored their far-reaching applications to quantum information theory. This enabled us to recover seemingly unrelated findings from  recent literature, like the strong subadditivity for log-determinant of CMs \cite{AdessoSerafini,Adesso} and  the basic properties of relevant measures of continuous variable entanglement \cite{AdessoSerafini} and steering \cite{steerability,MonSteer}, and to reach substantially beyond.  In particular, we proved that the Gaussian steerability \cite{steerability,JOSAB} for Gaussian states is a convex monotone under Gaussian local operations and classical communication, i.e., it is a fully fledged steering measure \cite{resource} within the Gaussianity restriction; we further proved it is monogamous with respect to the steering party for any (even non-Gaussian) multimode state, but not with respect to the steered party if the latter has more than one mode and the overall state is mixed. We also proved that the Gaussian R\'enyi-2 measure of entanglement \cite{AdessoSerafini} is monogamous for any Gaussian state with an arbitrary number of modes per party. This key result is a simple corollary of a general hierarchical relation here established for measures of correlations based on log-determinant of CMs.

This work further reveals how pursuing {\it prima facie} technical advances in classical information theory and linear algebra can significantly impact on the identification of possibilities and limitations for quantum technologies, which had eluded a general quantitative analysis so far.
It will be worth investigating adaptations of our results to the study of quantum correlations in discrete variable stabilizer states, useful resources for quantum computing \cite{rauss} which share deep mathematical analogies with continuous variable Gaussian states \cite{Gross06,Gross}.

\paragraph*{Acknowledgments.}
We warmly thank R.~Simon for many fruitful discussion on the topic of this work. We acknowledge financial support from the European Union under the European Research Council (StG GQCOP No.~637352 and AdG IRQUAT No.~267386) and the the European Commission (STREP RAQUEL No.~FP7-ICT-2013-C-323970),
the Foundational Questions Institute (fqxi.org) Physics of the Observer Programme (Grant No.~FQXi-RFP-1601),
the Spanish MINECO (Project No.~FIS2013-40627-P and FPI Grant No.~BES-2014-068888), and the Generalitat de Catalunya (CIRIT Project No.~2014~SGR~966).

\bibliographystyle{apsrevfixed}
\bibliography{schurbib}


\clearpage
\setcounter{equation}{0}
\begin{widetext}
\appendix

\section*{Appendix: Technical proofs and additional remarks}

\section{Properties of the functions $\boldsymbol{g_\pm}$} \label{app prop g}

Throughout this section, we study the functions $g_\pm$ introduced via \eqref{func g} and demonstrate their properties as stated in the main text. A decisive ingredient of our analysis is a version of the Courant-Fischer-Weyl variational principle for symplectic eigenvalues proven in the recent paper \cite{bhatia15}.
First of all, observe that $g_\pm(A)=g_\pm(SAS^T)$ for all symplectic $S$, since those functions are defined only in terms of symplectic eigenvalues. Moreover, it is immediately verified that
\begin{equation} g_\pm (A^{-1}) = g_\mp (A) \, , \qquad g_+(A)-g_-(A)=\frac{1}{2}\,\log\det A\, ,\qquad g_\pm (A\oplus B)=g_\pm (A) + g_\pm (B) \, . \label{elem prop g} \end{equation}
Perhaps less trivially, the following holds.

\begin{prop} \label{g- convex}
The function $g_-(A)$ is monotonically decreasing and convex in $A$, while $g_+$ is monotonically increasing but neither convex nor concave.
\end{prop}

\begin{proof}
The fact that $g_+,g_-$ are monotone in their inputs can be seen as an easy consequence of the symplectic equivalent of Weyl's monotonicity theorem first proven as Lemma 2 in \cite{sympsplemma} (and reported as Theorem 8.15 in \cite{Gosson}). That result states that if $A\geq B>0$ are $2n\times 2n$ real matrices, then their ordered symplectic eigenvalues satisfy $\nu_i (A)\geq \nu_i(B)$ for all $i=1,\ldots, n$. The claim follows by performing elementary manipulations.

Now, let us prove the convexity of $g_-(A)$. Our proof employs the recently found variational expression
\begin{equation}
\prod_{i=1}^k\, \nu_i(A)\ = \min_{S:\, S^T\Omega_{2n} S=\Omega_{2k}} \sqrt{\det(S^T A S)} \label{var expr}
\end{equation}
for the product of the $k$ smallest symplectic eigenvalues (see Theorem 5 in \cite{bhatia15}). In the above expression, $\Omega_{2k}$ denotes the standard symplectic form on $k$ modes. We easily find
\begin{equation}
g_-(A)\ =\  \max_{1\leq k\leq n}\, \sum_{i=1}^k (-\log \nu_i(A))\ =\ -\, \frac{1}{2}\, \min_{\scriptsize \begin{array}{cc} 1\leq k\leq n\, , \\ S:\, S^T\Omega_{2n} S=\Omega_{2k} \end{array} } \log\det(S^T A S)
\label{var expr g}
\end{equation}
Since $\log \det$ is well-known to be concave \cite{logconcave}, and $F(x)\equiv\min_{y\in Y} f(x,y)$ is always concave in $x$ if $f(x,y)$ was concave in $x$ for all fixed $y\in Y$, we infer that $g_-$ is indeed convex. Finally, in order to see that $g_+$ is neither convex nor concave it suffices to test it on positive multiples of the identity.
\end{proof}

\begin{rem} \emph{Why Proposition \ref{g- convex} does not imply that the logarithmic negativity is convex}. The formula $E_N = g_-(\tilde{V}_{AB})$, the linearity of the partial transposition $V_{AB}\mapsto \tilde{V}_{AB}$ on CMs and the convexity of $g_-$ could lead us to think that the logarithmic negativity is convex in the input state, which is false \cite{VidalWerner,plenioprl}. The reason why this chain of implications is not correct is that $E_N$ is expressible in terms of the CM only for Gaussian states, that do not constitute a convex set. However, it is true that if $\{\rho_i\}_i$ is a family of Gaussian states such that their convex combination $\sum_i p_i \rho_i$ is again Gaussian, then
\begin{equation} E_N \left(\sum_i p_i \rho_i\right)\, \leq\ \sum_i p_i E_N(\rho_i)\, . \end{equation}
One could call this behaviour \emph{Gaussian--convexity}. The logarithmic negativity is an example of a Gaussian--convex function which is in general non-convex.
\end{rem}

Proposition \ref{g- convex} can be used to prove that $g_-(V_{AB})$ decreases if the coherences between subsystems $A$ and $B$ are erased.

\begin{prop}[Decoherence reduces $g_-$] \label{dec red g-} $ \\ $
Let $V_{AB}>0$ be a bipartite positive definite matrix (not necessarily a quantum CM). Then
\begin{equation}
g_-(V_{AB})\ \geq\ g_-(V_A)\, +\, g_-(V_B)\, . \label{dec red g- eq S}
\end{equation}
\end{prop}

\begin{proof}
We will give a straightforward proof based on the convexity of $g_-$, but an alternative argument can be deduced directly from the variational expression \eqref{var expr g}. Observing that $\mathds{1}_A\oplus (-\mathds{1}_B)$ is a symplectic operation one finds
\begin{equation}
g_-(V_{AB})\ =\ g_-\Big( \big(\mathds{1}_A\oplus (-\mathds{1}_B) \big) \, V_{AB}\, \big(\mathds{1}_A\oplus (-\mathds{1}_B) \big) \Big)\, ,
\end{equation}
from which we infer
\begin{align}
g_-(V_{AB})\ &=\ \frac{1}{2} \left( g_-(V_{AB})\ +\ g_-\Big( \big(\mathds{1}_A\oplus (-\mathds{1}_B) \big) \, V_{AB}\, \big(\mathds{1}_A\oplus (-\mathds{1}_B) \big) \Big) \right)\ \geq\ g_-\left(\, \frac{1}{2}\, V_{AB}\, +\, \frac{1}{2}\, \big(\mathds{1}_A\oplus (-\mathds{1}_B) \big) \, V_{AB}\, \big(\mathds{1}_A\oplus (-\mathds{1}_B) \big)\, \right)\  \nonumber \\
&=\ g_-(V_A\oplus V_B)\ =\ g_-(V_A)\, +\, g_-(V_B)\, .
\end{align}
\end{proof}

\begin{rem}
One could be tempted to conjecture inequalities linking $g_-(V_{AB})$ with $g_-(V_A)$ and $g_-(V_{AB}/V_A)$. However, in general on the one hand $g_-(V_{AB})\ngeq g_-(V_A) + g_-(V_{AB}/V_A)$ (counterexample: bipartite quantum system $AB$  which is $A\rightarrow B$ steerable) and on the other hand $g_-(V_{AB})\nleq g_-(V_A) + g_-(V_{AB}/V_A)$ (there exist numerical counterexamples to that).
\end{rem}

We have seen that the function $g_-$ admits a variational expression \eqref{var expr g} in terms of a maximum (or the negative of a minimum). Now, we explore an alternative variational principle for $g_-$ (and $g_+$) that yields it directly as a minimum instead.

\begin{lemma} \label{second var g+-}
The functions $g_\pm$ admit the following representations:
\begin{align*}
g_+(W)\, &=\, \min_{W\leq\, Z\,\geq i\Omega} \ \frac{1}{2} \log\det Z\, ,\\
g_-(W)\, &=\, \min_{\Omega^T W^{-1}\Omega \leq \, Z\, \geq i\Omega} \ \frac{1}{2} \log\det Z\, .
\end{align*}
\end{lemma}

\begin{proof}
Since $g_+$ is increasing, clearly $W\leq Z$ implies $g_+(W)\leq g_+(Z)=\frac{1}{2}\log\det Z$, where the last equality holds because $Z$ is a quantum CM. This shows that $g_+(W) \leq \min_{W\leq\, Z\,\geq i\Omega} \ \frac{1}{2} \log\det Z$. On the other hand, Williamson's form $W=S\left(\begin{smallmatrix} \nu & 0 \\ 0 & \nu \end{smallmatrix}\right)S^T$ allows us to construct the ansatz $\bar{Z}\equiv S\left(\begin{smallmatrix} \bar{\nu} & 0 \\ 0 & \bar{\nu} \end{smallmatrix}\right) S^T$, where $\bar{\nu}_i \equiv \max\{ \nu_i, 1 \}$, which satisfies
\begin{equation}
W\leq \bar{Z}\geq i\Omega\, ,\qquad \frac{1}{2}\log\det \bar{Z} = \sum_i \max\{ 0,\, \log\nu_i(W)\} = g_+(W)\, .
\end{equation}
The expression for $g_-$ can be deduced from the one for $g_+$ with the help of the formula $g_-(W)=g_+(W^{-1})=g_+(\Omega^T W^{-1} \Omega)$.
\end{proof}

\begin{rem}
We remind the reader that for any $W>0$ the condition $Z \geq \Omega^T W^{-1} \Omega$ is equivalent to
\begin{equation}
\begin{pmatrix} W & \Omega \\ \Omega^T & Z \end{pmatrix} \geq 0\, .
\end{equation}
\end{rem}

\section{Properties of the steerability measure (Theorems \ref{G prop} and \ref{G mono})} \label{app prop steer}

We are now able to prove the physically fundamental properties of the Gaussian steerability measure \eqref{G steer} as stated in Theorem \ref{G prop}. In \cite{steerability}, (some of) these facts were stated and proven only in the particular case in which the steered system is made of one mode.

\begin{proof}[Proof of Theorem \ref{G prop}] $ \\[-3ex] $
\begin{itemize}
\item[(1)] {\it $\mathcal{G}(A\rangle B)_V$ is convex and decreasing as a function of the CM $V_{AB}>0$.} \\[0.5ex]
Both properties follow straightforwardly by combining concavity and monotonicity of the Schur complement with Proposition \ref{g- convex}. Let us prove convexity for instance. Since the Schur complement is concave, for any $V_{AB},W_{AB}>0$ and $0\leq p\leq 1$ we obtain
\begin{equation*}
(pV_{AB}+(1-p)W_{AB}) \big/ (pV_A+(1-p) W_A)\ \geq\ p\, V_{AB}/V_A + (1-p)\, W_{AB}/W_A\, .
\end{equation*}
Applying the fact that $g_-$ is decreasing and convex gives
\begin{align}
\mathcal{G}(A\rangle B)_{pV_{AB}+(1-p)W_{AB}}\ &=\ g_- \big( (pV_{AB}+(1-p)W_{AB}) \big/ (pV_A+(1-p) W_A) \big)\ \leq\ g_- \big( p V_{AB}/V_A + (1-p) W_{AB}/W_A\big)\ \\
&\leq\ p\, g_-(V_{AB}/V_A)\, +\, (1-p)\, g_- (W_{AB}/W_A)\ =\ p\, \mathcal{G}(A\rangle B)_V\, +\, (1-p)\, \mathcal{G}(A\rangle B)_W \, . \nonumber
\end{align}

\item[(2)] {\it ${\cal G}(A\rangle B)$ is additive under tensor products, i.e. ${\cal G}(A_1 A_2\rangle B_1 B_2)_{V_{A_1B_1}\oplus W_{A_2 B_2}} =  \mathcal{G}(A_1\rangle B_1)_{V_{A_1B_1}} +  \mathcal{G}(A_2\rangle B_2)_{W_{A_2 B_2}}$.} \\[0.5ex]
Elementary, since
\begin{align}
\mathcal{G}(A_1A_2\rangle B_1B_2)_{V_{A_1B_1}\oplus W_{A_2 B_2}}\ &=\ g_-\left( (V_{A_1B_1}\oplus W_{A_2 B_2}) \big/ (V_{A_1}\oplus W_{A_2}) \right)\, =\, g_-\left( V_{A_1B_1} / V_{A_1} \oplus W_{A_2 B_2}/ W_{A_2} \right)\, \\
&=\, g_-\left( V_{A_1B_1} / V_{A_1}\right)\, +\, g_-\left(W_{A_2 B_2} / W_{A_2} \right)\ =\, \mathcal{G}(A_1\rangle B_1)_{V_{A_1B_1}}\, +\, \mathcal{G}(A_2\rangle B_2)_{W_{A_2 B_2}}\, . \nonumber
\end{align}

\item[(3)] {\it For arbitrary states, $\mathcal{G}(A\rangle B)$ is monotonically decreasing under general, non-deterministic Gaussian maps on the steering party $A$.} \\[0.5ex]
Using the monotonicity of the Schur complement under general Gaussian maps, as given in Theorem \ref{CP incr sch} of the main text, one gets
\begin{equation*}
\Gamma_{A\rightarrow A'} (V_{AB}) \big/ \Gamma_{A\rightarrow A'}(V_A)\ \geq\ V_{AB} / V_A\, .
\end{equation*}
Applying $g_-$ to both sides yields exactly
\begin{equation}
\mathcal{G}(A'\rangle B)_{\Gamma_{A\rightarrow A'}(V_{AB})}\ \leq\ \mathcal{G}(A\rangle B)_{V_{AB}} \, .
\end{equation}

\item[(4)] {\it For Gaussian states, $\mathcal{G}(A\rangle B)$ is monotonically decreasing under general, non-deterministic Gaussian maps on the steered party $B$.} \\[0.5ex]
This is the most difficult claim to prove. First of all, we recall that any general, non-deterministic Gaussian map can always be obtained by: i) adding an uncorrelated ancillary system; ii) performing a global symplectic operation; and iii) measuring some of the modes by means of a Gaussian measurement \cite{nogo1,nogo2,nogo3}. Clearly, $\mathcal{G}(A\rangle B)$ is invariant under the addition of an ancillary steered system in an uncorrelated state because of the above point (2). Furthermore, the invariance under symplectic operations on $B$ is guaranteed by the very definition of $g_-$ in terms of symplectic eigenvalues. Thus, we are only left to prove that the Gaussian steerability decreases when a partial Gaussian measurement is performed on the steered system.

We remind the reader that a Gaussian measurement is comprised of a set of positive Gaussian operators obtained by applying displacement unitaries to a single positive Gaussian operator with quantum CM $\gamma$. It is known that, given a composite system $ABC$ in a Gaussian state with quantum CM $V_{ABC}$, when one measures the subsystem $C$ according to a Gaussian measurement with quantum CM $\gamma_C$, the reduced post-measurement state of subsystem $AB$ is Gaussian and with a quantum CM given by $\tilde{V}_{AB}=(V_{ABC}+\gamma_C)/(V_C+\gamma_C)$ (independently of the outcome). Bearing that in mind, we are claiming that for all quantum CMs $V_{ABC}$ one has
\begin{equation}
g_-\left( \tilde{V}_{AB} /\tilde{V}_A \right) \leq\, g_-(V_{ABC}/V_A)\, . \label{steerability decreases measurement}
\end{equation}
Call $W_{BC}\equiv V_{ABC}/V_A$. Then, a simple calculation that uses the quotient property of the Schur complement shows that
\begin{align*}
\tilde{V}_{AB} /\tilde{V}_A &= \big( (V_{ABC}+\gamma_C) / (V_C+\gamma_C) \big)\, \big/\, \big( (V_{AC}+\gamma_C) / (V_C+\gamma_C) \big) = (V_{ABC}+\gamma_C) / (V_{AC}+\gamma_C)\, \\
&=\, \big( (V_{ABC}+\gamma_C)/V_A \big)\, \big/\, \big( (V_{AC}+\gamma_C)/V_A \big) \, =\, \big( V_{ABC}/V_A + \gamma_C \big)\, \big/\, \big( V_{AC}/V_A+\gamma_C \big) \, =\, (W_{BC}+\gamma_C) / (W_C+\gamma_C)\, .
\end{align*}
Thus, \eqref{steerability decreases measurement} takes the form
\begin{equation}
g_-\left((W_{BC}+\gamma_C)/(W_C+\gamma_C)\right) \leq\, g_-\left( W_{BC} \right)\, , \label{steerability decreases measurement 2}
\end{equation}
to be proven for all $W_{BC}> 0$. Now, since the measured matrix $(W_{BC}+\gamma_C)/(W_C+\gamma_C)$ is concave in $\gamma_C$, and $g_-$ is decreasing and convex by Proposition \ref{g- convex}, we can restrict ourselves to prove inequality \eqref{steerability decreases measurement 2} only in the case in which $\gamma_C$ is symplectic, i.e., it is the CM of a {\it pure} Gaussian state.

Now we apply the above Lemma \ref{second var g+-} (together with the remark immediately below it). Suppose we found a matrix $Z_{BC}\geq i\Omega_{BC}$ such that
\begin{equation}
\begin{pmatrix} W_{BC} & \Omega_{BC} \\ \Omega_{BC}^T & Z_{BC} \end{pmatrix} \geq 0\, ,\qquad \frac{1}{2}\log\det Z_{BC} = g_-(W_{BC})\, . \label{global matrix 1}
\end{equation}
Then, consider the matrix
\begin{equation}
\begin{pmatrix} 0_B\oplus \gamma_C & 0_B \oplus \Omega_C^T \\[1ex] 0_B\oplus \Omega_C & 0_B \oplus \gamma_C \end{pmatrix} \geq 0\, , \label{global matrix 2}
\end{equation}
where the last inequality holds because $\gamma\geq \Omega^T\gamma^{-1}\Omega$ for all $\gamma\geq i\Omega$, as Williamson's decomposition immediately reveals. Adding \eqref{global matrix 2} to \eqref{global matrix 1} we get
\begin{equation}
\begin{pmatrix} W_{BC}+\gamma_C & \Omega_B\oplus 0_C \\[1ex] \Omega_B^T\oplus 0_C & Z_{BC}+\gamma_C \end{pmatrix} \geq 0\, . \label{global matrix 3}
\end{equation}
Taking the Schur complement with respect to the two $C$ components, thanks to the two crucial zero blocks we have just formed, we obtain
\begin{equation}
\begin{pmatrix} (W_{BC}+\gamma_C) / (W_C+\gamma_C) & \Omega_B \\[1ex] \Omega_B^T & (Z_{BC}+\gamma_C) / (Z_C+\gamma_C) \end{pmatrix} \geq 0\, .
\label{global matrix 4}
\end{equation}
Remarkably, since $Z_{BC}\geq i\Omega_{BC}$ one finds easily $(Z_{BC}+\gamma_C) / (Z_C+\gamma_C)\geq i\Omega_B$. Therefore, the same Lemma \ref{second var g+-} gives us
\begin{equation}
g_-\left( (W_{BC}+\gamma_C) / (W_C+\gamma_C) \right)\, \leq\, \frac{1}{2} \log \det (Z_{BC}+\gamma_C) / (Z_C+\gamma_C)\, .
\end{equation}
The proof is ended once we show that
\begin{equation}
\det (Z_{BC}+\gamma_C) / (Z_C+\gamma_C) \leq \det Z_{BC}\qquad \forall\ Z_{BC}\geq i\Omega_{BC},\ \forall\ \text{symplectic quantum CMs $\gamma_C$.}
\label{determinant decreases pure POVM}
\end{equation}
This rather surprising fact is hard to prove at the level of CMs, but it becomes more tractable once we come back to the Hilbert space picture behind. This can be done thanks to the identity $\Tr\rho_G^2=1\big/\!\sqrt{\det V}$, relating the {\it purity} $\Tr\rho_G^2$ of a Gaussian state $\rho_G$ to the determinant of its CM $V$. Such an identity allows us to restate \eqref{determinant decreases pure POVM} as the claim that {\it purity of Gaussian states increases when pure Gaussian measurements are applied.}

Let us now translate also the measurement into the Hilbert space picture. Since $\gamma_C$ is pure, the Gaussian measurement will be represented by a collection of rank-one (unnormalized) operators $\{ \ket{\psi_x}\!\!\bra{\psi_x}_C\}$ such that $\int d^{2n_C} x \ket{\psi_x}\!\!\bra{\psi_x}_C = \mathds{1}_C$. Furthermore, we have seen that the outcomes of this measurement on a Gaussian state $\rho_{BC}$ with covariance matrix $Z_{BC}$ will be states $\tilde{\rho}_B^{(x)}=U_x^\dag \tilde{\rho}_B U_x$, where $U_x$ are displacement unitaries depending on $x$ and $\tilde{\rho}_B$ is a Gaussian state independent of $x$ with CM $(Z_{BC}+\gamma_C)/(Z_C+\gamma_C)$. With these hypotheses, we now see that Lemma \ref{rk-1 POVM lemma} below allows us to conclude that $\Tr\tilde{\rho}_B^2 \geq \Tr \rho_{BC}^2$, which immediately yields \eqref{determinant decreases pure POVM} since both $\tilde{\rho}_B$ and $\rho_{BC}$ are Gaussian states.

\item[(5)] {\it The upper bound ${\cal G}(A\rangle C)_V  \leq  g_+(V_{BC}/V_B)$ holds for any quantum CM $V_{ABC}$ obeying \eqref{bonafide}.} \\[0.5ex]
Taking \eqref{INEQ 2}, applying $g_-$ and using the elementary properties \eqref{elem prop g} yields exactly
\begin{equation}
\mathcal{G}(A\rangle C)\ =\ g_-(V_{AC}/V_A)\ \leq\ g_+(V_{BC}/V_B)\, .
\end{equation}
\end{itemize}
\end{proof}

In proving point (4) of Theorem \ref{G prop} above, we used some unproven property of Gaussian pure measurements. We now clarify this point by stating two lemmas which complete the proof.

\begin{lemma} \label{x x_dag lemma}
Let
\begin{equation}
\begin{pmatrix} A & X \\ X^\dag & B \end{pmatrix} \geq 0
\end{equation}
be a hermitian, positive definite block matrix. Then
\begin{equation}
\|X\|_2^2 \leq \|A\|_2 \|B\|_2\, ,
\label{x x_dag eq}
\end{equation}
where $\|M\|_2=\sqrt{\Tr M^\dag M}$ denotes the Hilbert--Schmidt norm.
\end{lemma}

\begin{proof}
Using Cauchy--Schwartz inequality for the Hilbert--Schmidt product, the inequality $X^\dag A^{-1}X\leq B$, and the fact that the Hilbert--Schmidt norm is an increasing function on positive matrices, we obtain
\begin{align}
\|X\|_2^2\, &=\, \Tr XX^\dag\, =\, \Tr A\, A^{-1/2} XX^\dag A^{-1/2}\, \leq\, \|A\|_2\, \| A^{-1/2} XX^\dag A^{-1/2}\|_2\,  \nonumber \\
&=\, \|A\|_2\, \sqrt{\Tr A^{-1/2} XX^\dag A^{-1} XX^\dag A^{-1/2}}\, =\, \|A\|_2\, \sqrt{\Tr \left(X^\dag A^{-1} X \right)^2}\, \\
&=\, \|A\|_2\, \|X^\dag A^{-1} X\|_2\, \leq\, \|A\|_2\, \|B\|_2\, . \nonumber
\end{align}
\end{proof}

\begin{lemma} \label{rk-1 POVM lemma}
Suppose that the outcomes of the partial, rank-one measurement $\{ \ket{\psi_i}\!\!\bra{\psi_i}_C\}$ on a bipartite system $BC$ in a state $\rho_{BC}$ are always unitarily equivalent to a fixed density operator on the remaining system $B$, i.e.
\begin{equation}
{}_C \langle \psi_i | \rho_{BC} | \psi_i \rangle_C\, =\, p_i\, U_i^\dag \tilde{\rho}_B U_i\qquad \forall\ i\, .
\end{equation}
Then
\begin{equation}
\text{\emph{Tr}}\, \tilde{\rho}_B^2\, \geq\, \text{\emph{Tr}}\, \rho_{BC}^2\, . \label{rk-1 POVM eq}
\end{equation}
\end{lemma}

\begin{proof}
For the sake of brevity, in what follows we suppose that $i$ is an index running over a finite alphabet, but the argument below extends straightforwardly to the more general case in which it belongs to a measurable space.
Since the identity
\begin{equation}
\tilde{\rho}_B\, =\, U_i  \frac{{}_C \langle \psi_i | \rho_{BC} | \psi_i \rangle_C  }{p_i} U_i^\dagger\, ,
\end{equation}
is valid for all indices $i$, we obtain
\begin{equation}
\Tr \tilde{\rho}_B^2\, =\, \|\tilde{\rho}_B\|_2^2\, =\, \bigg( \sum_i p_i \,\|\tilde{\rho}_B\|_2 \bigg)^2\, =\, \bigg(\sum_i \, \big\| \langle \psi_i | \rho_{BC} | \psi_i \rangle \big\|_2\ \bigg)^2\, = \, \sum_{ij} \,\big\| \langle \psi_i | \rho_{BC} | \psi_i \rangle \big\|_2\, \big\| \langle \psi_j | \rho_{BC} | \psi_j \rangle \big\|_2\, , \label{rk-1 POVM eq1}
\end{equation}
where we omitted the subscript $C$ of $\ket{\psi_i}$ for the sake of brevity. Consider the map $\Phi_{C \rightarrow C'}$ from $C$ to a new system $C'$ such that
\begin{equation}
\Phi(X) = \bigg( \sum_i \ket{i}\!\!\bra{\psi_i} \bigg)\, X\, \bigg( \sum_j \ket{j}\!\!\bra{\psi_j} \bigg)^\dag = \sum_{ij} \ket{i}\!\!\bra{\psi_i} X \ket{\psi_j}\!\!\bra{j}\, .
\end{equation}
Obviously, $\Phi$ is completely positive and trace-preserving. Therefore, $\left(I_B\otimes \Phi_{C\rightarrow C'}\right)(\rho_{BC})\geq 0$ is a legitimate quantum state. This latter density matrix has blocks indexed by $i,j$ and given by ${}_C \langle \psi_i | \rho_{BC} | \psi_j \rangle_C$. Thanks to Lemma \ref{x x_dag lemma}, we know that for all $i\neq j$,
\begin{equation}
\big\| \langle \psi_i | \rho_{BC} | \psi_j \rangle \big\|_2^2\, \leq\, \big\| \langle \psi_i | \rho_{BC} | \psi_i \rangle \big\|_2 \, \big\| \langle \psi_j | \rho_{BC} | \psi_j \rangle \big\|_2\, .
\end{equation}
This is also trivially true when $i=j$. Therefore,
\begin{equation}
\sum_{ij} \,\big\| \langle \psi_i | \rho_{BC} | \psi_i \rangle \big\|_2\, \big\| \langle \psi_j | \rho_{BC} | \psi_j \rangle \big\|_2\, \geq\, \sum_{ij} \,\big\| \langle \psi_i | \rho_{BC} | \psi_j \rangle \big\|_2^2\, =\, \sum_{ij} \Tr \left[ \bra{\psi_i} \rho_{BC} \ket{\psi_j}\!\!\bra{\psi_j} \rho_{BC} \ket{\psi_i} \right]\, =\, \Tr \rho_{BC}^2\, , \label{rk-1 POVM eq2}
\end{equation}
where we used the normalization condition $\sum_i \ket{\psi_i}\!\!\bra{\psi_i}_C = \mathds{1}_C$. Inserting \eqref{rk-1 POVM eq2} into \eqref{rk-1 POVM eq1} yields the claim \eqref{rk-1 POVM eq}.
\end{proof}

\begin{rem}
\red{According to the resource theory of steering \cite{resource}, any valid quantifier of   ``$A$ to $B$'' steerability should be mandatorily (i) vanishing on unsteerable assemblages and (ii) nonincreasing on average under one-way LOCC (from $B$ to $A$), and optionally (iii) convex. A quantity satisfying (i), (ii), and (iii) is referred to as a {\it convex steering monotone}; examples of such monotones are discussed in \cite{resource}. In point (ii), one-way LOCC are defined as (ii.a) arbitrary
deterministic classical operations on the steering party $A$ and (ii.b) arbitrary non-deterministic quantum operations on the steered
party $B$ \cite{resource}.    Theorem \ref{G prop} proves  that, for an arbitrary bipartite Gaussian state with quantum CM $V_{AB}$, the measure ${\cal G}(A\rangle B)$ is convex and monotonic under arbitrary non-deterministic Gaussian quantum operations on either the steered or the steering party, which is even stronger than what required by the specialization of (ii) to Gaussian states and operations. Therefore  Theorem \ref{G prop} proves that  ${\cal G}(A\rangle B)$ is a valid convex steering monotone within the Gaussian subtheory of steering, settling a question left open in \cite{steerability,resource}. We further remark that parts (1)--(3) and (5) of Theorem \ref{G prop} hold for arbitrary continuous variable states, not necessarily Gaussian. On the other hand, our current proof of the monotonicity of ${\cal G}(A \rangle B)$ under general, non-deterministic Gaussian maps on $B$ relies on the specification to Gaussian states of $AB$. We leave it as an open problem whether a more general proof of part (4) could be obtained, valid even for non-Gaussian states.
}
\end{rem}

\begin{rem}
As a corollary of Theorem \ref{G prop}, one sees easily that the R\'{e}nyi--2 measure of entanglement
\begin{equation}
\mathcal{E}_2(A:B)\, =\, \inf_{\gamma_{AB}\, \text{pure}:\ \gamma_{AB}\,\leq\, V_{AB}}\, \frac{1}{2}\, \log\det \gamma_A\, ,
\end{equation}
is an upper bound on the steerabilities $\mathcal{G}(A\rangle B)$ and $\mathcal{G}(B\rangle A)$. In fact, consider the optimal pure $\gamma_{AB}\leq V_{AB}$ in the above equation and write
\begin{equation}
\mathcal{E}_2(A:B)_V\, =\, \frac{1}{2}\log\det \gamma_A\, =\, g_-(\gamma_{AB}/\gamma_A)\, \geq\, g_-(V_{AB}/V_A)\, =\, \mathcal{G}(A\rangle B)_V\, ,
\end{equation}
where we used first the expression of the steerability in terms of local determinant for pure states and then the fact that $\mathcal{G}({A \rangle B})$ is monotonically decreasing as a function of the CM.
\end{rem}

Now, let us discuss claim (b) of Theorem \ref{G mono}. As stated in the main text, the validity of \eqref{mon steer 2} for pure Gaussian states can be easily inferred by putting together inequality \eqref{E mono} and Theorem \ref{I2E}:
\begin{equation}
\mathcal{G}(B_1\ldots B_k \rangle A)_V\, =\, \mathcal{E}_2(B_1\ldots B_k:A)_V\, \geq\, \sum_{j=1}^k \mathcal{E}_2(B_j:A)_V \, \geq\, \sum_{j=1}^k \mathcal{G}(B_j\rangle A)_V\, ,
\end{equation}
where the first equality holds specifically for pure states.

On the contrary, already in the simplest case $k=2$, $n_A=2,\, n_{B_1}=n_{B_2}=1$, there exist mixed states violating inequality~\eqref{mon steer 2}. A counterexample is as follows:
\begin{equation} V_{AB_1B_2}\ =\ \begin{pmatrix}
 1.2 & -0.3 & 0.4 & -2.7 & 1.8 & -1.9 & 0.4 & -0.1 \\
 -0.3 & 0.9 & -1.2 & 0.4 & -1.2 & 0.5 & -0.4 & 0.1 \\
 0.4 & -1.2 & 4.5 & 1.6 & -1.4 & 1.8 & -0.1 & -0.3 \\
 -2.7 & 0.4 & 1.6 & 12. & -9.5 & 10.1 & -1.4 & -0.3 \\
 1.8 & -1.2 & -1.4 & -9.5 & 11.9 & -11.5 & 1.6 & 0.8 \\
 -1.9 & 0.5 & 1.8 & 10.1 & -11.5 & 11.9 & -1. & -1.4 \\
 0.4 & -0.4 & -0.1 & -1.4 & 1.6 & -1. & 2.4 & -2. \\
 -0.1 & 0.1 & -0.3 & -0.3 & 0.8 & -1.4 & -2. & 2.8
\end{pmatrix}\, .
\end{equation}
Here, the first four rows and columns pertain to $A$, the fifth and sixth to $B_1$, the last two to $B_2$. It can be easily verified that the minimum symplectic eigenvalue of the above matrix with respect to the symplectic form $\Omega_A\oplus \Omega_{B_1}\oplus\Omega_{B_2}$ is $\nu_{\min}(V_{B_1B_2A})=1.01359$, so that $V_{B_1B_2 A}$ is a legitimate quantum CM (obeying \eqref{bonafide}). However,
\begin{equation}
\mathcal{G}(B_1B_2\rangle A)_V\, -\, \mathcal{G}(B_1\rangle A)_V\, -\, \mathcal{G}(B_2\rangle A)_V\, =\, -0.816863\, .
\end{equation}

\section{R\'enyi-2 correlation hierarchy (Theorem \ref{I2E})} \label{app hierarchy}

This section is devoted to prove one of the main results of the paper, i.e. the correlation hierarchy for Gaussian states of arbitrarily many modes that is the content of Theorem \ref{I2E}. A fundamental tool we employ is the {\it geometric mean} between two positive matrices $M,N>0$ \cite{Ando}. For an excellent introduction to this topic, we refer the reader to chapter 4 of \cite{Bhatia}. We limit ourselves to discuss (without proofs) the main properties of this remarkable quantity.

The geometric mean $M\# N$  of two matrices $M,N>0$ can be defined in many equivalent ways:
\begin{itemize}
\item it is the only function of $M,N>0$ which is invariant under simultaneous congruence (that is, $(LML^T)\# (LNL^T)=L (M\# N) L^T$) and reduces to $\sqrt{MN}$ when $[M,N]=0$;
\item $M\# N\, \equiv\, M^{1/2}\left( M^{-1/2} N M^{-1/2} \right)^{1/2} M^{1/2}$;
\item $M\# N\, \equiv\, \max\, \left\{ X=X^T: \begin{pmatrix} M & X \\ X & N \end{pmatrix} \geq 0 \right\}\, =\, \max\, \{ X=X^T: M\geq XN^{-1}X\} $;
\item $M\# N$ is the unique positive definite solution of the Riccati equation $M=XN^{-1} X$ with unknown $X$.
\end{itemize}
Furthermore, this special matrix function enjoys many desirable properties:
\begin{itemize}
\item $M\# N=N\#M$;
\item $M\# N$ is jointly concave in $M,N>0$;
\item $\Phi(M\# N)\leq \Phi(M)\# \Phi(N)$ for all positive maps $\Phi$.
\end{itemize}

The following lemma is our first result.

\begin{lemma}
Let $V$ be a quantum CM obeying \eqref{bonafide}. Then $\gamma_V^\# \equiv V\#(\Omega V^{-1}\Omega^T)$ satisfies $\gamma_V^\# \leq V$ and is a quantum CM of a pure Gaussian state (i.e.~it is symplectic).
\end{lemma}

\begin{proof}
Consider Williamson's decomposition $V=S\left( \begin{smallmatrix} D & 0 \\ 0 & D \end{smallmatrix} \right) S^T$, where $D\geq \mathds{1}$ is diagonal and $S$ is symplectic \cite{Williamson36}. One finds
\begin{equation*}
\Omega V^{-1}\Omega^T\ =\ \Omega S^{-T} \begin{pmatrix} D^{-1} & 0 \\ 0 & D^{-1} \end{pmatrix} S^{-1} \Omega^T\ =\ S \Omega \begin{pmatrix} D^{-1} & 0 \\ 0 & D^{-1} \end{pmatrix} \Omega^T S^T\ =\ S \begin{pmatrix} D^{-1} & 0 \\ 0 & D^{-1} \end{pmatrix} S^T\, .
\end{equation*}
Since the geometric mean is covariant under congruence (see the above discussion or Corollary 2.1 of \cite{Ando}), one finds
\begin{equation*}
\gamma_V^\#\ =\ V\#(\Omega V^{-1}\Omega^T)\ =\ S \left(\begin{pmatrix} D & 0 \\ 0 & D \end{pmatrix} \# \begin{pmatrix} D^{-1} & 0 \\ 0 & D^{-1} \end{pmatrix}\right) S^T\ =\ SS^T\, .
\end{equation*}
Both claims easily follow.
\end{proof}

\begin{proof}[Proof of Theorem \ref{I2E}]
Thanks to the above discussion, we have only to prove that ${\cal I}_2\geq 2\mathcal{E}_2$ holds for all Gaussian states. Using the pure state of the above lemma as ansatz in the definition \eqref{E2} of the R\'enyi-2 entanglement measure, and denoting by $\Pi_A$ the projector onto the $A$ component, one sees that
\begin{equation*}
\mathcal{E}_2(A:B)_V\, \leq\, \frac{1}{2}\, \log\det (\gamma_{V_{AB}}^\#)_A\, =\, \frac{1}{2}\, \log\det \left( \Pi_A \left(V_{AB}\#(\Omega_{AB} V_{AB}^{-1} \Omega_{AB}^T)\right) \Pi_A^T\right)\, .
\end{equation*}
Employing the inequality $\Phi(M\# N)\leq \Phi(M)\# \Phi(N)$ (Theorem 3 of \cite{Ando}) with the positive map $\Phi(\cdot)=\Pi_A(\cdot)\Pi_A^T$ we find
\begin{equation*}
\Pi_A \left(V_{AB}\#(\Omega_{AB} V_{AB}^{-1} \Omega_{AB}^T)\right) \Pi_A^T\ \leq\ V_A \# \left( \Pi_A (\Omega_{AB} V_{AB}^{-1} \Omega_{AB}^T) \Pi_A^T \right)\ =\ V_A\# \left( \Omega_A (V_{AB}/V_B)^{-1} \Omega_A^T \right)\, ,
\end{equation*}
where for the last step we used the well-known formula \eqref{inv} for the inverse of a $2\times 2$ block matrix. Inserting this operator inequality into the above upper bound for $\mathcal{E}_2(A:B)_V$ we obtain
\begin{equation*}
\mathcal{E}_2(A:B)_V\, \leq\, \frac{1}{2}\, \log\det \left( V_A\# \left( \Omega_A (V_{AB}/V_B)^{-1} \Omega_A^T \right)\right)\, =\, \frac{1}{4}\,\log\frac{\det V_A}{\det V_{AB}/V_B}\, =\, \frac{1}{4}\, \log\frac{\det V_A \det V_B}{\det V_{AB}}\, =\, \frac{1}{2}\, {\cal I}_2(A:B)_V\, .
\end{equation*}
\end{proof}

\begin{rem}
Putting together Eqs.~(14) and (17) in \cite{AdessoSerafini} we deduce the weaker inequality $\mathcal{E}_2\leq {\cal I}_2$ (proven there only for two--mode Gaussian states).
\end{rem}

\begin{rem}
What makes the R\'enyi-2 entropy special in the context of the above proof? It turns out that for all R\'enyi-$\alpha$ entropies with $\alpha\geq 1$ (included the von Neumann one) we can always provide the upper bound
\begin{equation}
2\ \mathcal{E}_\alpha(A:B)_V\ \leq\ {\cal S}_\alpha \left( V_A\# \left( \Omega_A (V_{AB}/V_B)^{-1} \Omega_A^T \right)\right)\, ,
\end{equation}
where the function ${\cal S}_\alpha(V)$ gives the R\'enyi-$\alpha$ entropy of a Gaussian state with quantum CM $V$, i.e.
\begin{equation}
{\cal S}_\alpha (V)\, =\, -\frac{1}{\alpha-1}\, \sum_{i=1}^n \log \frac{2^\alpha}{\big(\nu_i(V)+1\big)^\alpha-\big(\nu_i(V)-1\big)^\alpha}
\label{S_alpha}
\end{equation}
for $\alpha>1$, and
\begin{equation}
{\cal S}_1(V)\, =\, \sum_{i=1}^n\left(\frac{\nu_i(V)+1}{2}\,\log\frac{\nu_i(V)+1}{2}\ -\ \frac{\nu_i(V)-1}{2}\,\log\frac{\nu_i(V)-1}{2}\right) \label{S_1}
\end{equation}
for the von Neumann case $\alpha=1$ (see for instance Eq.~(108) of \cite{AdessoReview}). However, the crucial inequality
\begin{equation}
{\cal S}_\alpha(M\# N)\,\leq\, \frac{1}{2} {\cal S}_\alpha(M)+\frac{1}{2} {\cal S}_\alpha (N) \label{ineq geom mean}
\end{equation}
breaks down for $\alpha<2$. In particular, it can be violated for $\alpha=1$. On the contrary, ${\cal I}_\alpha\geq 2\mathcal{E}_\alpha$ is always true as long as $\alpha\geq 2$, as the next Lemma clarifies.
\end{rem}

\begin{lemma} \label{lemma alpha}
Fix an integer $n\geq 1$. The inequality ${\cal S}_\alpha(M\# N)\leq \frac{1}{2} {\cal S}_\alpha(M)+\frac{1}{2} {\cal S}_\alpha (N)$ holds for all $2n\times 2n$ real matrices $M,N>0$ if and only if $\alpha\geq 2$.
\end{lemma}

\begin{proof}
We claim that inequality \eqref{ineq geom mean} is equivalent to the convexity of the function
\begin{equation}
f_\alpha (x)\, \equiv\, -\frac{1}{\alpha-1}\,\log\frac{2^\alpha}{(e^x+1)^\alpha - (e^x-1)^\alpha} \label{f_alpha}
\end{equation}
defined on $\mathds{R}_+$, where conformally to \eqref{S_1} one defines
\begin{equation}
f_1 (x)\, \equiv\, \frac{e^x+1}{2}\,\log\frac{e^x+1}{2}\ -\ \frac{e^x-1}{2}\,\log\frac{e^x-1}{2}\, . \label{f_1}
\end{equation}
In fact, on the one hand choosing $M=e^x\mathds{1},\, N=e^y \mathds{1}$ yields
\begin{align*}
{\cal S}_\alpha(M\# N) &= n\, f_\alpha\left(\frac{x+y}{2}\right)\, , \\
\frac{1}{2} {\cal S}_\alpha(M)+\frac{1}{2} {\cal S}_\alpha (N)\, &=\, \frac{n}{2} f_\alpha (x) + \frac{n}{2} f_\alpha (y)\, ,
\end{align*}
so that $f_\alpha$ is necessarily convex when \eqref{ineq geom mean} holds. On the other hand, suppose that $f_\alpha$ is convex. From Theorem 3 of \cite{bhatia15} we learn that $\log\hat{\nu}(M\# N)\prec \frac{1}{2}\log\hat{\nu}(M) +\frac{1}{2}\log\hat{\nu}(N)$, where $\hat{\nu}(M)\in\mathds{R}^{2n}_+$ is obtained by listing the symplectic eigenvalues of $M$ each repeated twice and sorting the entries of the resulting vector in descending order, the logarithm of vectors is intended entrywise, and the symbol $\prec$ denotes {\it majorization} (see Chapter II of \cite{BhatiaMatrix}). What the above relation tells us is that the symplectic spectrum of the geometric mean is in a precise sense more disordered than the geometric mean of the two spectra. It is elementary to verify that whenever $x\prec y$ and $f$ is convex, $\sum_{i=1}^n f(x_i)\leq \sum_{i=1}^n f(y_i)$ holds true (see Corollary II.3.4 of \cite{BhatiaMatrix}). Choosing as $f$ the function in \eqref{f_alpha} and observing that $f_\alpha$ is always monotonically increasing, we obtain
\begin{align}
{\cal S}_\alpha(M\# N)\, &=\, \sum_{i=1}^n f_\alpha\left( \log \nu_i(M\# N) \right)\, =\, \frac{1}{2}\sum_{i=1}^{2n} f_\alpha\left( \log \hat{\nu}_i(M\# N) \right)\, \leq \, \frac{1}{2}\sum_{i=1}^{2n} f_\alpha\left( \frac{1}{2}\log\hat{\nu}_i(M) +\frac{1}{2}\log\hat{\nu}_i(N) \right)  \\ \, &\leq \, \frac{1}{2}\sum_{i=1}^{2n} \frac{1}{2} \left( f_\alpha(\log\hat{\nu}_i(M)) + f_\alpha\left(\log\hat{\nu}_i(N)\right) \right)\, =
 \, \frac{1}{2} {\cal S}_\alpha(M)+\frac{1}{2} {\cal S}_\alpha (N)\, . \nonumber
\end{align}

Now, the main claim will follow once we show that $f_\alpha$ defined via \eqref{f_alpha} is convex if and only if $\alpha\geq 2$. We can restrict our analysis to the case $\alpha>1$ since the function in \eqref{f_1} is elementarily seen to be non-convex (actually, concave). Some tedious algebra leads us to the following expression for the second derivative of that function:
\begin{equation}
f_\alpha''(x)\ =\ \frac{\alpha}{\alpha-1}\, \frac{\cosh^\alpha(x/2) \sinh^\alpha(x/2)}{\sinh^2 x\, (\cosh^\alpha(x/2) - \sinh^\alpha(x/2))}\, \left( \cosh^\alpha(x/2) \sinh^{2-\alpha}(x/2)\, -\, \sinh^\alpha (x/2) \cosh^{2-\alpha}(x/2)\, -\, \alpha +1 \right)\, .
\label{f_alpha''}
\end{equation}
Since everything else in the above expression is positive, we have only to prove that $\cosh^\alpha (x/2) \sinh^{2-\alpha} (x/2)\, -\, \sinh^\alpha (x/2) \cosh^{2-\alpha} (x/2)\, \geq \, \alpha -1$ for all $x\geq 0$ if and only if $\alpha\geq 2$. That $\alpha\geq 2$ is necessary can be seen by taking the limit $x\rightarrow 0^+$. Conversely, if $\alpha=2+\delta$ with $\delta\geq 0$ one gets
\begin{align}
\cosh^\alpha (x/2) \sinh^{2-\alpha} (x/2)\, -\, \sinh^\alpha (x/2)\cosh^{2-\alpha} (x/2)\, &=\, \cosh^2 (x/2)\, \tanh^{-\delta} (x/2)\, -\, \sinh^2 (x/2)\, \tanh^\delta (x/2)\, \nonumber \\
&=\, \frac{1}{1-t^2}\, t^{-\delta} \,-\, \frac{t^2}{1-t^2}\, t^{\delta}\, \equiv\, \varphi_t(\delta)\, ,
\end{align}
where we defined $t\equiv \tanh(x/2)$. It is not difficult to see that the function $\varphi_t(\delta)$ is convex in $\delta$ since
\begin{equation}
\varphi''_t(\delta)\, =\, \frac{t^{-\delta } \left(1-t^{2 \delta +2}\right) \log ^2(t)}{1-t^2}\, \geq\, 0\, .
\end{equation}
From this fact we deduce that
\begin{equation}
\varphi_t(\delta)\, \geq\, \varphi_t(0)+\delta\, \varphi'_t(0)\, =\, 1 + \delta\, =\, \alpha-1\, ,
\end{equation}
as claimed.
\end{proof}

\clearpage
\end{widetext}

\end{document}